\documentclass[review,12pt]{elsarticle}
\usepackage[utf8]{inputenc}
\usepackage{xcolor}
\usepackage{hyperref}
\usepackage{amsmath,amsfonts,amssymb,amsthm}
\usepackage{mdwlist}
\usepackage{fullpage}
\usepackage{graphicx}
\usepackage{placeins} 
\graphicspath{{./figures/}}

\newtheorem{thm}{Theorem}
\newtheorem{lem}[thm]{Lemma}
\newtheorem{prop}[thm]{Proposition}
\newtheorem{cor}[thm]{Corollary}
\newtheorem{rem}[thm]{Remark}
\newtheorem{ques}[thm]{Question}
\newtheorem{exmp}[thm]{Example}

\theoremstyle{definition}
\newtheorem{definition}[thm]{Definition}

\newcommand\W{\mathcal W}
\newcommand\U{\mathcal U}
\newcommand\T{\mathcal T}

\newcommand\K{\mathcal K}
\newcommand\ct[1][G]{\mathcal C_{#1}}

\newcommand{\actson}{\curvearrowright}

\newcommand\deux{\mathbf2}
\newcommand{\R}{\mathbb R}
\newcommand{\Z}{\mathbb Z}
\newcommand{\N}{\mathbb N}
\renewcommand{\H}{\mathbb H}
\newcommand\orb{\omega}

\newcommand{\PGL}{\text{PGL}}
\newcommand{\PSL}{\text{PSL}}
\newcommand{\Isom}{\text{Isom}}
\newcommand{\sym}{\text{Sym}}

\newcommand{\substitution}{\sigma}
\newcommand{\tiling}{\mathbf{T}}
\newcommand{\pattern}{M}
\newcommand{\support}{\mathrm{supp}}
\usepackage{ifthen}
\newcommand{\ifnv}[2]{\ifthenelse{\equal{#1}{}}{}{#2}}
\newcommand{\sett}[3][]{\left\{\left.#2\ifnv{#1}{\in #1}\vphantom{#3}\right|#3\right\}} 
\newcommand\Stab{\text{Stab}}
\newcommand\St[1]{\Stab(#1)}
\newcommand{\defeq}{:=}

\newcommand{\dfn}[1]{\textbf{#1}}
\DeclareMathOperator*{\id}{id}
\newcommand\grp[1][0]{\Psi_{#1}}
\newcommand\geo[1][0]{\orb_{#1}}
\newcommand{\resp}[1]{\ (resp. #1)}

\newcommand\gat{\text{Cucaracha}} 

\journal{Theoretical Computer Science}

\begin{document}

\begin{frontmatter}
\title{Aperiodic monotiles: from geometry to groups}
\author[AMU]{Thierry Coulbois}  
\author[AMU,UdeC]{Anahí Gajardo}
\author[AMU]{Pierre Guillon}
\author[AMU]{Victor Lutfalla}
\affiliation[AMU]{organization={Institut de Mathématique de Marseille (I2M), CNRS, Aix-Marseille Université},
           country={France}}
\affiliation[UdeC]{organization={Departamento de Ingeniería Matemática, Center for Research on Mathematical Engineering (CI$^2$MA), Universidad de Concepción},
           country={Chile}}

\begin{abstract}
In 2023, two striking, nearly simultaneous, mathematical discoveries have excited their respective communities, one by Greenfeld and Tao, the other (the \emph{Hat} tile) by Smith, Myers, Kaplan and Goodman-Strauss, which can both be summed up as the following: there exists a single tile that tiles, but not periodically (sometimes dubbed the \emph{einstein} problem).
The two settings and the tools are quite different (as emphasized by their almost disjoint bibliographies): one in Euclidean geometry, the other in group theory.
Both are highly nontrivial: in the first case, one allows complex shapes; in the second one, also the space to tile may be complex.

We propose here a framework that embeds both of these problems.
  From any tile system in this general framework, with some natural additional conditions, we exhibit a construction to \emph{simulate} it by a group-theoretical tiling.
  We illustrate this by transforming the \emph{Hat} tile into a new aperiodic group monotile,
  and we describe the symmetries of both the geometrical \emph{Hat} tilings and the group tilings we obtain. 
\end{abstract}




\end{frontmatter}

\section{Introduction}

Tilings were originally defined as coverings of the Euclidean plane by compact tiles without overlap. 
They have been studied since the ancient Greeks, with a special emphasis on periodic tilings in which a finite domain is repeated periodically on the whole plane.

Seminal results on tilings include the classification of Archimedean tilings (tilings by regular polygons where all vertices are the same up to isometry) by Kepler~\cite{kepler} and the study of the symmetry groups of periodic tilings (crystallographic space groups) by Bieberbach~\cite{bieberbach}.

In the 60s emerged the question of \emph{aperiodic} tilesets~\cite{wang},
that is, sets of tile that can tile the whole plane but only in a non-periodic way.
Though Wang initially conjectured such tilesets did not exist, Berger~\cite{berger} soon provided a first aperiodic tileset of 20426 tiles up to translation.

\begin{figure}[p]
  \center
  \includegraphics[width=\textwidth]{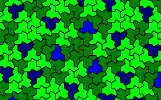}
  \includegraphics[width=\textwidth]{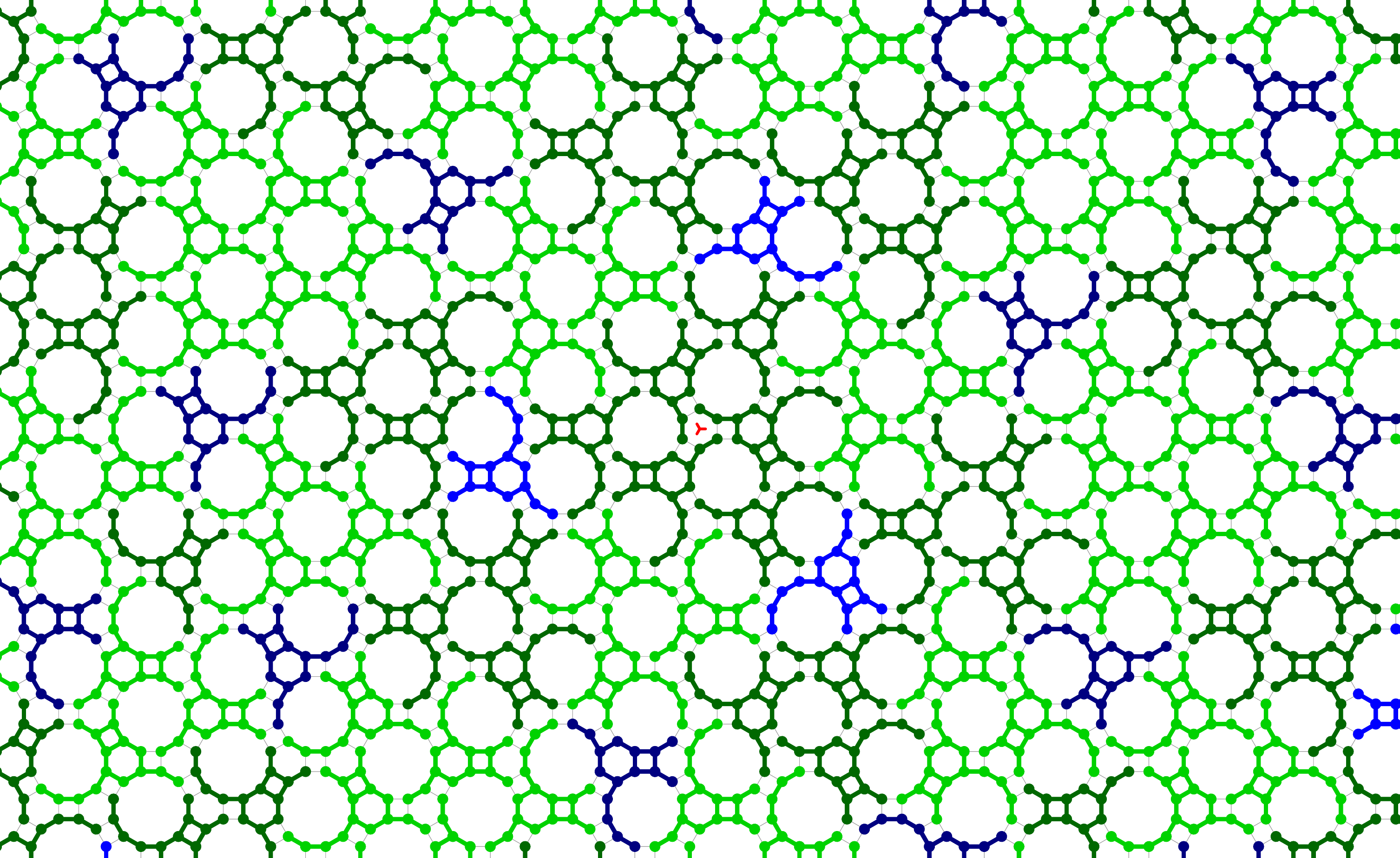}
  \caption{A Hat tiling of $\R^2$ and a corresponding \gat~tiling on group $\Gamma$.
    The colors correspond to signature of the symmetry, and the darkness correspond to the parity of angle modulo $\pi/6$.
    The red tricross is a center of $3$-fold rotational symmetry.}
  \label{fig:discrete_hat_tiling}
\end{figure}

The first simple geometric aperiodic tileset was found by Robinson in 1971~\cite{robi71}, who defined a set of 6 geometric tiles up to isometry with the shape of squares with bumps and dents, that can tile the 2-dimensional plane, but without any translational symmetry.  
This started the challenge that was taken two years later by
Penrose~\cite{Penr74}, who defined an aperiodic set of only $2$ geometric tiles.
Penrose's tiling has fivefold rotational symmetry, and, unlike Robinson's, it cannot be arranged in a square grid.

In 1992, Ammann, Grünbaum and Shephard~\cite{polyominoes} restricted the problem to polyominoes, and found an aperiodic set of $3$ polyominoes, if one allows rotations, or $8$ otherwise.
{
  The latter record lasted until 2024, when Yang and Zhang~\cite{polyo7} found an aperiodic set of $7$ polyominoes.
On the other hand, in 1991, Beauquier and Nivat~\cite{bnivat} showed, using methods from word combinatorics, that (if no rotation of the tiles is allowed) no aperiodic connected polyomino monotile exists.

This polyomino tiling problem can be expressed in terms of \emph{group tiling} of $\Z^2$:
a polyomino is indeed a finite connected subset of $\Z^2$, and tiling amounts to covering the group $\Z^2$ with translations of such subsets. 

While geometric tilings are classically considered in $\R^2$, group tilings may be defined in any group, which changes the techniques that can be applied to tackle the problem.
The problem of tiling $\Z^2$ with translations of only one tile was eventually solved in 2020 by Bhattacharya~\cite{bhattacharya}, who proved that no finite subset of $\Z^2$ (possibly disconnected in the usual Cayley graph) is an aperiodic monotile.
On the other hand, in 2023 by
Greenfeld and
Tao~\cite{gtao2} constructed an aperiodic monotile in a group of the form $\Z^2\times H$, for a well-chosen finite abelian group $H$.
Thanks to the lifting result by \cite{lift}, they deduce that such a monotile also exists for $\Z^d$, when $d$ is big enough.

Coming back to geometric tilings, the tiling by Penrose kept the minimality
record for 50 years, unless one authorizes tiles to be disconnected, to communicate through their vertices, or to have some thickness, see~\cite{socolar}.
In 2023 finally, Smith, Myers, Kaplan and
Goodman-Strauss~\cite{hat} described new aperiodic tilings with
only $1$ tile. The work is very interesting, not only because a
whole family of aperiodic monotiles was discovered (with a fundamentally original aperiodicity argument), but also because
two of these tiles (the Hat and the Turtle) are polykites and can be
embedded in a \emph{Kitegrid}, whose dual is an Archimedean
tiling (which can be seen as the Cayley graph of a group), suggesting
that the Hat and the Turtle tilings can also be seen in some sense as group tilings.

Elaborating on this idea, we generalize here the notion of polyominoes into poly-$K$, which are finite unions of copies of some given compact $K$ which tiles the plane periodically.
Polyominoes correspond to $K$ being a unit square, while other classical examples
include polyamonds ($K$ is the unit equilateral triangle) or polykites ($K$ is the Kite).
}

We then provide a toolkit to translate geometric tilings by a
poly-$K$ tile into tilings of the symmetry group of the $K$-grid.
In particular, we apply this toolkit to the Hat polykite \cite{hat} to
obtain a new aperiodic monotile (which we call \gat) in a finitely presented group $\Gamma$ which is virtually $\Z^2$, see
Figure~\ref{fig:discrete_hat_tiling}.
It can be compared to \cite[Theorem~1.4]{gtao2}: our result is somehow weaker in the sense that the extension of $\Z^2$ is not direct.
On the other hand, it is more explicit, and enjoys a nice Coxeter definition, as well as nice pictures.

Section~\ref{sec:tilings} provides terminology and definitions for geometric and group tilings.
Section~\ref{s:discr} presents our toolkit on translation from poly-$K$ tiles to group tilings.
Section~\ref{sec:hat} applies this toolkit to the Hat tile.
Section~\ref{sec:general} discusses two variations on our main result, that link geometric and group tilings outside of the poly-$K$ case.

The readers interested in abstract tilings and the link between geometric
and group tilings may skip Section~\ref{sec:hat}, while readers only
interested in the \gat\ tile can use Corollary~\ref{c:main-one-k} as a 
blackbox and focus on Section~\ref{sec:hat}.

\section{Tilings} \label{sec:tilings}
Let $\W$ be a topological space endowed with a nontrivial Borel measure $\lambda$.
Let $G$ be a group which acts by self-homeomorphisms $g$ of $\W$ that we assume \dfn{negligibility-preserving} (that is, $\lambda(A)>0\implies\lambda(g\cdot A)>0$).
A \dfn{tile} is a positive-measure subset of $\W$ (often considered up to negligible sets).
A \dfn{cotiler} for a set $\T$ of tiles is a countable subset $C\subset G\times\T$ 
such that:
\[
  \bigsqcup_{(g,T)\in C}g\cdot T\equiv_\lambda\W,\]
where $\equiv_\lambda$ is an equality and $\sqcup$ a disjoint union, both up to $\lambda$-negligible sets.
If $\T$ admits a cotiler, we say that $\T$ \dfn{tiles} $\W$ with respect to $G$.

The terminology is inspired from the case when $\T$ is a singleton $\{T\}$; in that case we will say that $T$ tiles and confuse the cotiler with a subset of $G\sim G\times\{T\}$.

The cotiler should not be confused with the \dfn{tiling} $\sett{g\cdot T}{(g,T)\in C}$.

\begin{rem}\label{r:tiles}
  We keep a very large setting because very few assumptions are needed for our results (some more will be added to the tiles and the group in some results).
  Nevertheless, one can keep in mind that natural examples arise when $\W$ is a second-countable locally compact space, $\lambda$ has full support (is positive for every nonempty open subset), and the tiles $T$ are compact, are the closure of their interior, have finite positive measure and $\lambda$-negligible boundaries.
  In this setting, {every cotiler is countable and} $T\equiv_\lambda T'$ if and only if $T=T'$.
\end{rem}

Here are some classical examples of settings.
\begin{exmp}
Relevant examples of spaces $\W$ include the following:
  \begin{itemize}
\item
  $\W=\R^d$ endowed with canonical topology, Lebesgue measure, $G=\R^d$ acting by translations or $G = \Isom^+(\R^d)\simeq\R^d\rtimes SO_d(\R)$ acting by direct isometries or $G =
 \Isom(\R^d)\simeq \R^d\rtimes O_d(\R)$ acting by isometries. 
  We are also interested in tilings with respect to some countable subgroups of these.
  It is also natural to require additional constraints over the tiles, like polygonality\ldots
\item
  Consider a countable group $\W=G$
  endowed with discrete topology, the counting measure,
  self-acting by left multiplication.
  The notion of \emph{monotileable} group~\cite{monotile} and its generalization into \emph{polytileable} group~\cite{polytile} 
  fit in this framework.
    \item
     More generally, $\W = G$ is a locally compact group endowed with the Haar measure, self-acting by left multiplication.
\item
  $\W=\H^2$, the hyperbolic plane and $G=\Isom(\H^2)\simeq \PGL_2(\R)$ or $G=\Isom^+(\H^2)\simeq \PSL_2(\R)$.
\item Tilings of half-planes like $\N\times\Z$ or $\R_+\times\R$ could be covered by a suitable generalization of our theory to monoids, which would nevertheless have the price of a complexification of many notions.
\end{itemize}\end{exmp}

The simplest examples of tilings are periodic tilings, introduced through the following formalism.
\begin{definition}\label{d:grid}
  A finite tileset $\K$ tiles with cotiler $C=G\times\K$ if, and only if,
  $G$ is countable and $\bigsqcup_{K\in\K}K$ is, up to $\lambda$-negligible sets, a fundamental domain for $G\actson\W$.
  We then say that $\K$ \dfn{yields a grid}, in the sense that:
  \[\W=\bigsqcup_{(h,K)\in G\times\K}h\cdot K.\]
\end{definition}
This includes regular tilings like the square lattice, but also some with different tiles like hexagonal-square-triangle lattice (note that $\K$ can include several tiles from the same translation class).

\begin{rem}\label{r:ct}~\begin{enumerate}
  \item Any tiling with respect to some $G$ can always be considered as a tiling with respect to the group of all negligibility-preserving self-homeomorphisms. Nevertheless, our goal is to reduce $G$ as much as possible (while still allowing tilings).
  \item\label{i:ct} If $\T$ tiles with some cotiler $C$, then a relevant group of self-homeomorphisms of $\W$ is the (countable) \dfn{cotiler group} $G_C\defeq\langle gh^{-1}\ |\ (g,T),(h,T)\in C,T\in\T\rangle$.
    Indeed, for every $(g_T)_{T\in\T}$ such that $(g_T,T)\in C$, one can see that the tileset $\sett{g_T\cdot T}{T\in\T}$ (essentially equal to $\T$, up to translating each tile) tiles with cotiler $\{(gg_T^{-1},g_T\cdot T)\ |\ (g,T)\in C\}\subseteq G_C\times \sett{g_T\cdot T}{T\in\T}$.
  \end{enumerate}\end{rem}


\subsection{Periods}
In general, if $\T$ tiles, then it admits many cotilers.
Let us denote $\ct[G](\T)$ the set of its cotilers.

$G$ acts naturally on $\ct[G](\T)$ by $g\cdot C\defeq\sett{(gh,T)}{(h,T)\in C}$.
The set $\ct[G](\T)$ can be seen as some kind of \emph{subshift of finite type}
, which means that one can check that something is indeed a cotiler by looking at bounded windows: if we encode $C$ as the configuration $x\in(\deux^\T)^G$ (where $\deux^\T$ consists of the subsets of $\T$), defined by $x_g=\sett[\T]T{(g,T)\in C}$, we can check whether $C$ is a cotiler by checking $x$ on every $g$ and an appropriate neighborhood around $g$.


A \dfn{period} for a cotiler $C$ is some $h\in G$ such that $h C=C$.
The \dfn{stabilizer} of $C$ is the 
subgroup $\St C\le G$ of its periods (see Figure~\ref{fig:periods} for an example).
A cotiler $C$ is:
\begin{itemize}
\item \dfn{weakly periodic} if $\St C$ is nontrivial.
\item \dfn{mildly periodic} if $\St C$ is infinite.
\item \dfn{strongly periodic} if $\St C\actson\W$ cocompactly.
\end{itemize}
A tileset (which admits at least a tiling) is \dfn{strongly}\resp{\dfn{mildly}, \dfn{weakly}} \dfn{aperiodic} if no cotiler is weakly\resp{mildly, strongly} periodic.
Of course, strong periodicity implies mild periodicity, which implies weak periodicity, except in the case of finite groups (where all cotilers are strongly, but not mildly, periodic), which does not interest us.

\begin{rem}\label{r:mild}~\begin{enumerate}
  \item In the case when $\W=G$ is a discrete group, this notion of strong periodicity is equivalent to the classical definition that $\St C$ has finite index in $G$, see \cite{bitar}.
  \item In the case when $G$ is virtually $\Z^2$, classical pumping arguments give that a tileset is weakly aperiodic if, and only if, it is mildly aperiodic, see \cite{ballier}.
  \end{enumerate}\end{rem}

Note also that in the case when $G$ contains other isometries than translations, (a)periodicity differs from the usual notion of (a)periodicity with respect to the subgroup of translations.

\begin{center}
\begin{figure}[htb!]
 \includegraphics[width=.3\textwidth]{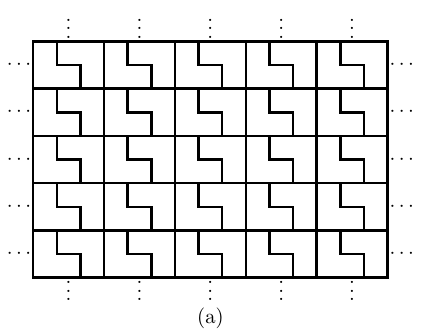}\hspace{.3cm}
 \includegraphics[width=.35\textwidth]{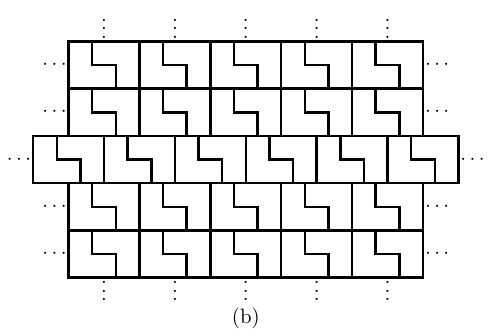}\hspace{.3cm}
 \includegraphics[width=.3\textwidth]{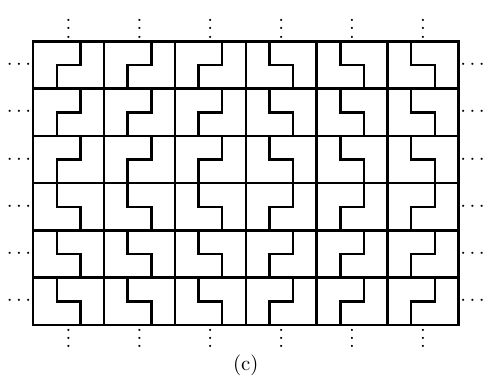}
 \caption{Here we consider a single {\bf L} shaped tile in $\mathbb{R}^2$ with its whole group of isometries, we present 3 different cotilers, 
   we denote by $t_{u}$ the translation by vector $u$, $R_2^u$, the rotation by $\pi$ at point $u$, and $\rho_h, \rho_v$ the horizontal and vertical reflections at the origin: (a) strongly periodic cotiler: $C=\{t_{(3a,2b)}\}_{(a,b)\in\mathbb{Z}^2}\cup\{R_2^{(3a,2b)}\circ t_{(3a,2b)}\}_{(a,b)\in\mathbb{Z}^2}$; (b) mildly but not strongly periodic cotiler: $\St C=\{t_{(3a,0)}\}_{a\in\mathbb{Z}}$; (c) weakly but not mildly periodic cotiler:\break$\St C=\{id,R_2^{(0,0)},\rho_h,\rho_v\}$.}\label{fig:periods}
\end{figure}
\end{center}

A classical metaquestion, in each world $\W$ (and with some constraints on $G$ and on the possible tiles), is whether there exists a finite aperiodic tileset.
In $\Z$ and $\R$, the answer is ``\emph{no}'' (folklore), and the proof carries to virtually-$\Z$ groups.
In $\R^2$ and $\Z^2$, the answer is ``\emph{yes}''.
The technique of turning colored tiles 
into shapes could be extended to prove that the answer is still ``\emph{yes}'' in $\R^d$ and $\Z^d$, for $d\ge 2$.
We are not aware of other groups for which this question has been studied (though extensive literature exists for sets of ``colored tiles''; see \cite{bitar} for a pleasant survey).

In cases where aperiodic tilesets do exist, it is then natural to look for the simplest ones; in particular, in the recent decades, the community has asked about the existence of an aperiodic monotile, that is, an aperiodic tile singleton.
In some settings no weakly aperiodic monotile exists, in particular:
\begin{itemize}
  \item in $\Z^2$ \cite{bhattacharya, gtao1};
  \item in $\R^2$ when $G$ is the group of translations and the tile is required to be a topological disk \cite{bnivat, bnivat2};
  \item in $\R^n$ when $G$ is the group of translations and the tile is required to be convex \cite{venkov, mcmullen};
  \item in $\R^2$ when $G$ is the group $\Isom(\R^2)$ of isometries of the plane and the tile is required to be a convex polygon \cite{rao}.
\end{itemize}
Nevertheless, there exist weakly aperiodic monotiles in some other settings:
\begin{itemize}
\item in $\R^2$ when $G$ is the group of isometries and the tile is connected \cite{hat};
\item in $\R^2$ when $G$ is the group of direct isometries (\cite{socolar} with a disconnected tile, or more recently \cite{spectre} with a connected tile);
\item in $\Z^2\times H$ for some finite abelian group $H$ \cite[Theorem~1.3]{gtao2};
\item in some $\Z^d$ for some large $d\in\N$ \cite[Corollary~1.5]{gtao2}.
\end{itemize}
There is a very tight boundary between group $\Z^2$ which does not admit aperiodic monotile, and some commensurable groups which do.
It is interesting to understand how small such a group can be.

\section{Discretization of poly-$\K$ tiles}\label{s:discr}

This section is devoted to the translation between two worlds: a geometric, continuous world, and a combinatorial, countable world.

\begin{rem}
  If $G$ is a group and $\K$ is any finite set, $G$ naturally acts on $G\times\K$ 
  by $g\cdot(h,K)\mapsto(gh,K)$.
  Tiling $G\times\K$ is equivalent to tiling $G$ with $\K$ layers that are synchronized (such considerations are used in \cite[Section~3]{gtao1}):
  If $\T\subset G\times\K$ is a tileset, and $K\in\K$, one can define the set $\T_K$ of 
  tiles ${\sett g{(g,K)\in T}}$, for ${T\in\T}$.
  Then the cotiler sets are related by $\ct(\T)=\bigcap_{K\in\K}\ct(\T_K)$.
\end{rem}

Let $\K$ be a finite set of tiles $K\subset\W$ that yields a grid with respect to $G$ in the sense of Definition~\ref{d:grid}, that is: $\K$ tiles $\W$ with $G$ as a cotiler for each tile. 

One says that a tile $T$ is a \dfn{poly-$\K$} if it is (up to $\lambda$-negligible sets) the disjoint union of finitely many copies of tiles from $\K$, that is,
\[T\equiv_\lambda\bigsqcup_{(h,K)\in\grp[\K](T)}h\cdot K,\]
for some finite $\grp[\K](T)\subset G\times\K$.
  For $K\in\K$, let us also denote $\grp[K](T)\defeq\grp[\K](T)\cap (G\times\{K\})$. 

\begin{thm}[Main result; poly-$\K$ version]\label{t:maink}
  Let $G$ be countable and $\K$ a finite set which yields a grid for $G\actson\W$.
  Then $\grp[\K]$ is a bijection from the set of ($\equiv_\lambda$-classes of) poly-$\K$ sets $\T$ of $\W$ onto the set of finite tilesets of
  $G\times\K$, such that the cotiler sets $\ct(\grp[\K](\T))$ and $\ct(\T)$ are equal, up to syntactically renaming each tile $T$ into $\grp[\K](\T)$.
\end{thm}
Formally, the last equality means 
Consequently, $\W$ admits (weakly, mildly, strongly) aperiodic poly-$\K$ sets (for the action of $G$) if, and only if, $G\times\K$ admits (weakly, mildly, strongly) aperiodic sets of finite tiles, with the same number of tiles.

Let us restate the particular case when $\K$ is a singleton $\{K\}$ (in that case we talk about poly-$K$) and $\T$ is a singleton $\{T\}$.
\begin{cor}\label{c:main-one-k}
  Let $G$ be countable and $K\subset\W$ yield a grid for $G\actson\W$.
  Then $\grp[K]$ is a bijection from the set of ($\equiv_\lambda$-classes of) poly-$K$ $T$ of $\W$ onto the set of finite tilesets of $G$, such that $\ct(\grp[K](T))=\ct(T)$, up to syntactically renaming $T$ into $\grp[\K](T)$.
\end{cor}
\begin{proof}[Proof of Theorem~\ref{t:maink}]
  The fact that $\grp[\K]$ is a bijection comes directly from the definition and that $h\cdot K$ and $h'\cdot K'$ do not intersect if $(h,K)\ne(h',K')$.
   
  $C$ is cotiler for $\T$ if, and only if:
  \begin{align*}
    &\iff&\W&\equiv_\lambda\bigsqcup_{(g,T)\in C}g\cdot T&&\text{ (definition of cotiler)}\\
    &\iff&\W&\equiv_\lambda\bigsqcup_{(g,T)\in C}\bigsqcup_{(h,K)\in\grp[\K](T)}g h\cdot K&&\text{ ($T$ is poly-$\K$)}\\
    &\iff&\bigsqcup_{(h,K)\in G\times\K}h\cdot K&\equiv_\lambda\bigsqcup_{(g,T)\in C}\bigsqcup_{(h,K)\in\grp[\K](T)}g h\cdot K&&\text{ ($\K$ yields a grid)}\\
    &\iff&\bigsqcup_{(h,K)\in G\times\K}\{(h,K)\}&=\bigsqcup_{(g,T)\in C}\bigsqcup_{(h,K)\in\grp[\K](T)}\{(g h,K)\}&&\text{($h\cdot K\cap h'\cdot K'\equiv_\lambda\emptyset$
)}\\
    &\iff&G\times\K&=\bigsqcup_{(g,T)\in C}g\cdot\grp[\K](T),
  \end{align*}
  which means that $\sett{(g,\grp[\K](T))}{(g,T)\in C}$ is a cotiler for tile $\grp[\K](T)$.
\end{proof}

\section{Application to the Hat}\label{sec:hat}

\subsection{The Hat}
The Hat is a polykite tile, hence we first define the kite and Kitegrid.

\begin{definition}[Kite and Kitegrid]
  We call \dfn{Kite} the quadrilateral from Figure~\ref{fig:kite-def}, which we consider as a compact subset of $\R^2$.
  We call \dfn{Kitegrid} the periodic Kite tiling which consists of the superposition of a hexagonal grid (with sides 2) and a triangular grid (with sides $2\sqrt{3}$), superimposed in such a way that the centers of the hexagons are the vertices of the triangles, see Figure~\ref{fig:kite-def}.

  We also call \dfn{Semikite} the right triangle obtained when bisecting the Kite along its long diagonal, and \dfn{Semikitegrid} the periodic Semikite tiling obtained from bisecting the Kites of the Kitegrid into Semikites, see Figure~\ref{fig:laves_archimedean} (right).
\end{definition}

\begin{figure}[htp]
  \center
  \includegraphics[width=0.4\textwidth]{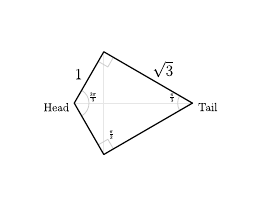} \hspace*{1cm}
  \includegraphics[width=0.4\textwidth]{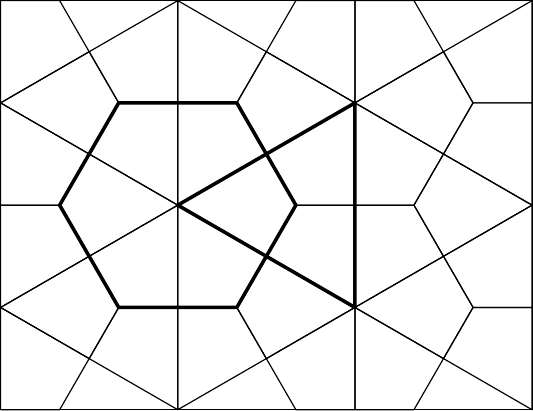}
  \caption{The Kite (left) and Kitegrid (right).}
  \label{fig:kite-def}
\end{figure}

We call \dfn{symmetry} of a tiling $\tiling$ an isometry of the plane that preserves the tiling, that is, $h\in \Isom(\R^2)$ is a symmetry of $\tiling = \{ g\cdot T | (g,T)\in C\}$ when $h\tiling=\tiling$. Denote $\sym(\tiling)$ the group of the symmetries of $\tiling$.
Similarly, we call \dfn{symmetry} of a tileset $\T$ an isometry of the plane that sends some tile $T\in\T$ to some tile $T'\in\T$ (possibly the same tile).

Let us emphasize that $h$ is a symmetry of a tiling $\tiling$ with cotiler $C$ when for  any $(g,T)\in C$ there exists $(g',T)\in C$ such that $h g\cdot T = g'\cdot T$ (here the equality is as subsets of $\W$), whereas $h$ is a period of $C$ when for any $(g,T) \in C$ there exists $(g',T)\in C$ such that $(h g,T)=(g',T)$, that is, $h g=g'$.
In particular, the symmetry group $\sym(\tiling)$ includes the stabilizer $\Stab(C)$ of the cotiler as a subgroup.
These two groups are not always equal: in particular for the Kitegrid, the stabilizer contains only direct isometries of the plane whereas the symmetry group also contains reflections.

However, if the tileset has no non-trivial symmetry itself, then $\Stab(C)=\sym(\tiling)$.



An \dfn{Archimedean tiling} is a tiling of $\R^2$ by regular polygons such that every vertex has the same configuration (the same sequence of polygons when reading around the vertex), see Figure~\ref{fig:laves_archimedean} in black.
An Archimedean tiling is characterized by this vertex configuration; for example the $4.3.4.6$ Archimedean tiling is a tiling by regular triangles, squares and hexagons such that around each vertex there is a square, a triangle, another square and a hexagon, see Figure~\ref{fig:laves_archimedean} (left).

The dual tilings of Archimedean tilings are called \dfn{Laves tiling}, see Figure~\ref{fig:laves_archimedean} in grey.
Laves tilings are characterized by the vertex configuration of their dual Archimedean tiling, and have only one tile up to isometry.
The Kitegrid is the $4.3.4.6$ Laves tiling, and the Semikitegrid is the $4.6.12$ Laves tiling.

\begin{figure}[htp]
  \center
  \includegraphics[width=0.4\textwidth]{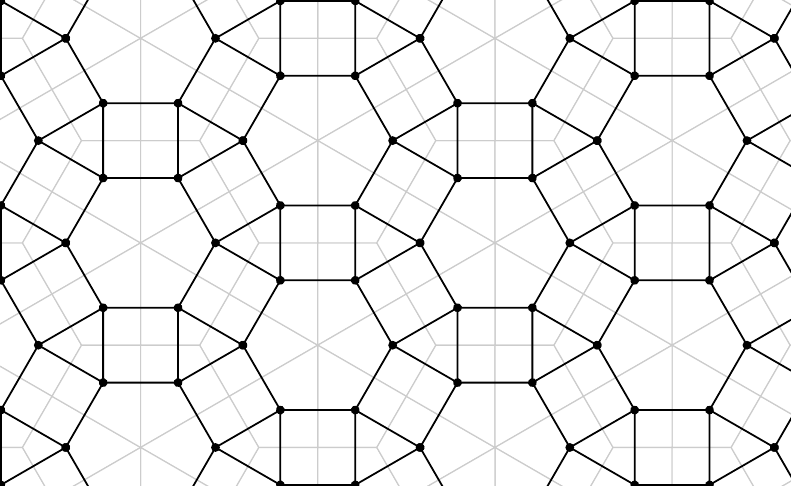} \hspace*{1cm}
  \includegraphics[width=0.4\textwidth]{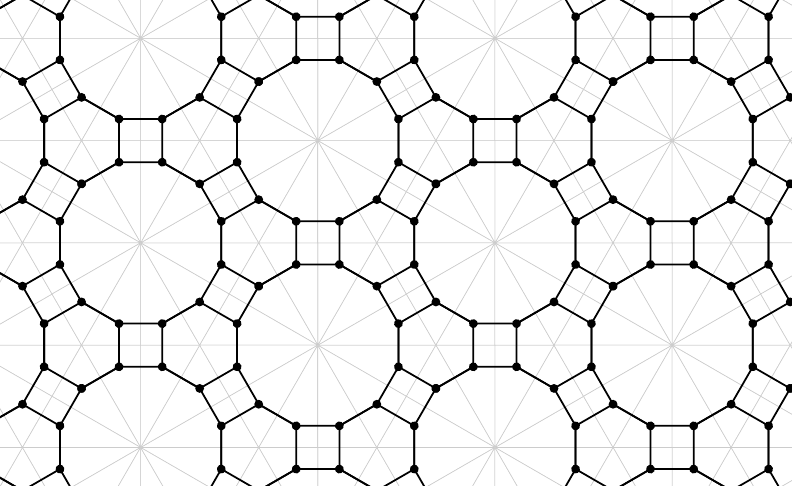}
  \caption{The Kitegrid (left) and Semikitegrid (right) in light grey, with their dual graph.\break
    The dual graphs induce the Archimedean tilings of vertex configurations $3.4.6.4$ (left) and $4.6.12$ (right).}
  \label{fig:laves_archimedean}
\end{figure}

\begin{definition}[Hat]
We call \dfn{Hat} the 13-gon tile represented in Figure~\ref{fig:hat-def} as a union of 8 kites.
\end{definition}

\begin{figure}[htp]
  \center
  \includegraphics[width=0.3\textwidth]{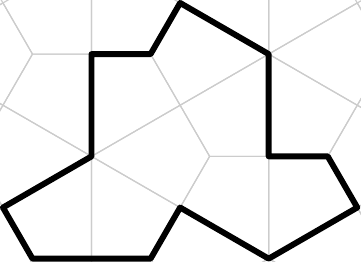}
  \caption{The Hat tile as a polykite.}
  \label{fig:hat-def}
\end{figure}

Note that we can either see the Hat as a 13-gon (6 edges of length $\sqrt{3}$, 6 edges of length $1$ and 1 edge of length $2$) or as a 14-gon with two adjacent colinear edges (6 edges of length $\sqrt{3}$ and 8 edges of length $1$).

We now consider tilings of $\W = \R^2$ with respect to the group $G$ of isometries of $\R^2$.

Among the main known results on the Hat tilings, we now present two that are useful in our construction.
\begin{thm}[Hat tilings \cite{hat}]~
  \label{thm:cite-hat}
  \begin{enumerate}
  \item\label{i:aperhat} The Hat is an aperiodic geometric monotile \cite[Theorem 1.1]{hat};
    in our terminology: the Hat tiles $\R^2$ with respect to $G=\Isom(\R^2)$;
    and no stabilizer of a cotiler for the Hat contains a translation,
    in particular it is weakly aperiodic.
  \item\label{i:kitehat} Decomposing any Hat tiling in kites yields the Kitegrid \cite[Lemma A.6]{hat};
    in our terminology: any cotiler group of the Hat is included in the symmetry group $\Gamma$ of the Kitegrid.
  \end{enumerate}
\end{thm}

\begin{rem}
  Though it is not formally stated in the original paper, the stabilizer of the Hat is finite, so that it is mildly aperiodic.
  This comes from the fact that $\Gamma$ is virtually $\Z^2$ (see Section~\ref{subsec:group_formalism}), thus any infinite subgroup contains a translation (see Remark~\ref{r:mild}).
\end{rem}

Section~\ref{subsec:group_formalism} presents the symmetry group $\Gamma$ of the Kitegrid;
Section~\ref{subsec:hat_symmetries} describes the possible stabilizers of Hat tilings; and Section~\ref{subsec:discretehat} introduces the \gat, a mildly aperiodic monotile in $\Gamma$ derived from the Hat.

Sections~\ref{subsec:group_formalism} and~\ref{subsec:hat_symmetries} are mostly independent.

\subsection{The symmetry group of the Kitegrid}
\label{subsec:group_formalism}

In this Section, we present the symmetry group of the Kitegrid, which we denote $\Gamma$. The readers who want to avoid some group formalism can refer to the Cayley graph from Figure~\ref{fig:gamma_cayley}.

Let $R_6$ be the rotation of angle $\pi/3$ centered at the tail of
(any fixed translate of) the Kite and $R_3$ be the rotation of angle $2\pi/3$ centered at
the head of the Kite. Let $\Gamma^{+}\defeq\langle R_6,R_3\rangle$ be the
subgroup of $\Isom^+(\R^2)$ (the group of orientation-preserving isometries of $\R^2$) generated by these two
rotations. $\Gamma^{+}$ is the cotiler of the Kite $K$ yielding the
Kitegrid, or equivalently, $K$ is a fundamental domain for the action of $\Gamma^{+}$ (up to $\lambda$-negligible sets):
\[\W=\R^2\equiv_\lambda\bigsqcup_{g\in\Gamma^{+}}g\cdot K.\]
We remark that $t_1=R_6^{-1}R_3R_6^{-1}$ and $t_2=R_6^{-2}R_3$ are two
translations of the plane that generate a lattice
$L=\langle t_1,t_2\rangle\simeq\Z^2$. The regular hexagon
$\bigsqcup_{i=0}^5R_6^i\cdot K$ is a fundamental domain for the group $L$
acting on the plane.

We also remark that the equilateral triangle
$\bigsqcup_{i=0}^3R_3^i\cdot K$ is a fundamental domain for the action of
$\Gamma'
\defeq L\rtimes\langle R_6^3\rangle=\langle t_1,t_2,R_6^3\rangle\simeq\Z^2\rtimes\Z/2\Z$.

For the sake of completeness, we remark that $\Gamma^{+} = L\rtimes\langle R_6\rangle=\Gamma'\rtimes\langle R_6^2\rangle=\Gamma'\rtimes\langle R_3\rangle\simeq\Z^2\rtimes\Z/6\Z$.

We now add the reflection $\beta$ along the axis of the Kite going from the tail to the head.
We get a group $\Gamma\defeq \Gamma^{+}\rtimes\langle\beta\rangle=\langle R_6,R_3,\beta\rangle$
which is the cotiler of the Semikite $K_+$ yielding the Semikitegrid;
equivalently, $K_+$ is a fundamental domain for the action of $\Gamma$ (up to $\lambda$-negligible sets).
The reflections $\alpha=R_6^{-1}\circ\beta$ and $\gamma=R_3\circ\beta$ along the other two sides of the Semikite (see Figure~\ref{fig:gamma-symmetries}) are also in the group $\Gamma$.

\begin{rem}
  The group $\Gamma^{+}$ of the Kitegrid acts freely and transitively on the vertices of the Archimedean tiling 4.3.4.6;
  it acts freely on edges, with two orbits: one orbit for the sides of the hexagons, one orbit for the sides of the triangles.
 Labelling the edges around hexagons by $R_6$ and the edges around triangles by $R_3$, we get that the $1$-skeleton of the 4.3.4.6 Archimedean tiling is the Cayley graph of $\Gamma^{+}$ with respect to generators $\{R_6,R_3\}$.

 The group $\Gamma$ of the Semikitegrid acts freely and transitively on the vertices of the Archimedean tiling 4.6.12.
 Labelling the edges of the Archimedean tiling 4.6.12 by $\alpha$, $\beta$ and $\gamma$ according to the edge of the Semikitegrid that they cross, we get that the $1$-skeleton of the Archimedean tiling 4.6.12 is the Cayley graph of $\Gamma$ with respect to generators $\{\alpha,\beta,\gamma\}$ (with non-oriented edges, because $\alpha$, $\beta$ and $\gamma$ are involutions).
\end{rem}


The symmetry group $\Gamma$ of the Semikitegrid has presentation:
\[\Gamma=\langle\alpha,\beta,\gamma|\alpha^2,\beta^2,\gamma^2,(\alpha\beta)^6,(\beta\gamma)^3,(\alpha\gamma)^2\rangle,\]
which is the presentation of a Coxeter group.

Remark that $\Gamma$ is also the symmetry group of the Kitegrid. 

Its index-$2$ subgroup $\Gamma^{+}$ has presentation $\langle R_3,R_6|R_3^3,R_6^6,(R_3R_6)^2\rangle$
with $R_3=\gamma\beta$ and $R_6=\beta\alpha$.

Figure~\ref{fig:gamma_cayley} shows the Cayley graph of the group $\Gamma$.

\begin{figure}[htp]
  \center
  \includegraphics[width=0.6\textwidth]{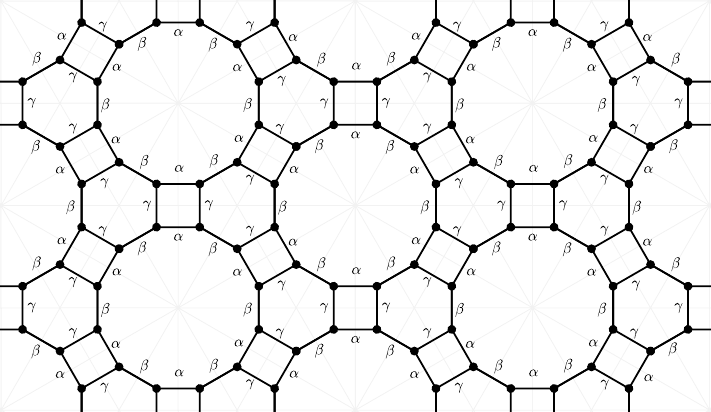}
  \caption{ The Cayley graph of $\Gamma$. 
  As the generators $\alpha$, $\beta$ and $\gamma$ are involutions, the edges are not directed.}
  \label{fig:gamma_cayley}
\end{figure}

In two senses, these groups are not far from $\Z^2$.
  \begin{rem}
    $\Gamma^{+}$ and $\Gamma$ are virtually $\Z^2$ and both admit planar Cayley graphs
    (see Figures~\ref{fig:laves_archimedean} and~\ref{fig:gamma_cayley}). 
  \end{rem}

  \begin{figure}[htp]
    \center
    \includegraphics[width=0.5\textwidth]{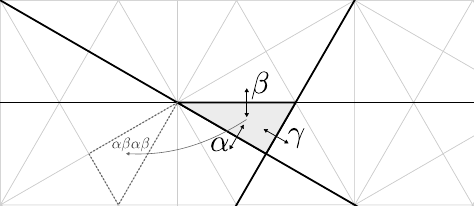}
    \caption{The symmetries of the Semikitegrid.\\
      Any semikite in the tiling is the image of the grey semikite by a composition of the reflections $\alpha$, $\beta$ and $\gamma$.\break
    In particular the semikite in dotted lines is the image of the grey semikite by $\alpha\beta\alpha\beta$.}
    \label{fig:gamma-symmetries}
  \end{figure}

\subsection{Stabilizers of Hat tilings}
\label{subsec:hat_symmetries}

In this section, we study symmetries and periods of the Hat tilings.
We chose to make these statements and proofs as independent as possible from the formalism of Section~\ref{s:discr} so that a reader might choose to avoid most of the group-action formalism.

Recall that since the Hat tiles has no symmetries, for any Hat tiling $\tiling$ with cotiler $C \subset \Isom(\R^2)$ we have $\Stab(C)=\sym(\tiling)$.

Recall that $R_3$ is the rotation of angle $\tfrac{2\pi}{3}$ around the origin head of the Kite, that we assume to be on the origin $0\in\R^2$.

\begin{prop}\label{p:aperhat}
  The symmetry group of any Hat tiling is either $\{\id\}$ or conjugated to $\{\id,R_3,R_3^2\}$.  
\end{prop}
\begin{proof}
  From Item~\ref{i:kitehat} of Theorem~\ref{thm:cite-hat}, any Hat tiling yields the Kitegrid when decomposing the tiles in kites, so that its symmetries must be included in the symmetry group $\Gamma$ of the Kitegrid.

  Let us analyse the subgroups of $\Gamma$ which could be candidates as stabilizers of Hat tilings:
  \begin{itemize}
  \item From Item~\ref{i:aperhat} of Theorem~\ref{thm:cite-hat},
  no Hat tiling admits a translation as a symmetry.
  \item no reflection is possible: indeed the axis of a reflection symmetry of a Hat tiling cannot intersect the interior of any tile.
    As any Hat tiling induces a Kitegrid when decomposing the tiles, such an axis would be a line in the Kitegrid.
    There are only two such lines on the boundary of a Hat (which do not intersect its interior), both of which are obstructed by a neighbouring tile, see Figure~\ref{fig:hat-no-reflection}.
    Hence a Hat tiling does not admit reflection symmetries.
    \begin{figure}[htp]
      \center
      \includegraphics[width=0.4\textwidth]{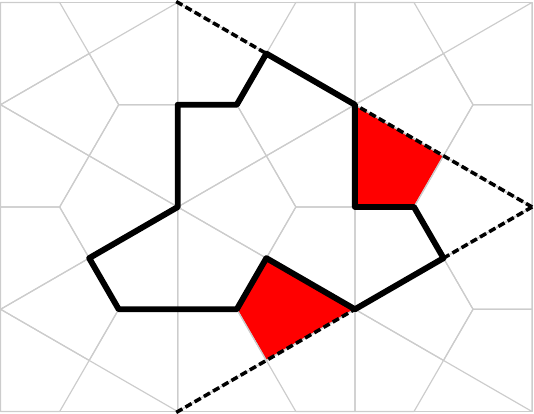}
      \caption{A Hat tile and its two non-intersecting adjacent KiteGrid axes (dashed), in red the obstructing kites: any neighbour Hat tile containing a red kite would intersect the corresponding axis.}
      \label{fig:hat-no-reflection}
    \end{figure}
  \item no nontrivial sliding reflection is possible: otherwise iterating it twice would give a translation symmetry.
  \item rotations of order 6 are not possible either: indeed the Hat tile does not have a symmetry of order 6, so a center of symmetry would have to be a vertex; this it is impossible because the Hat has only angles of $\frac{\pi}{2}$ and $\frac{2\pi}{3}$.
  \end{itemize}
  This leaves only the rotations of order $3$ as possible symmetries for Hat tilings.
  Every rotation of order $3$ in $\Gamma$ is conjugated to $R_3$ by some translation of $\Gamma$.
\end{proof}

This result gives an upper bound on the stabilizers and symmetry groups of Hat tilings.

We now prove that this bound is tight, that is: there exists a Hat tiling $\tiling = (\text{Hat}, C)$ such that $\sym(\tiling) = \Stab(C) = \{\id,R_3,R_3^2\}$.

As a tool to generate tilings, an auxiliary tileset of 4 polysemikites, called $HTPF$, was introduced in \cite{hat}, see Figure~\ref{fig:htpf-clusters} in dotted lines, and a substitution on this tileset, see Figure~\ref{fig:htpf-subst}.
Note that the $H$, $T$, $P$ and $F$ tiles have an arrow as label, which is an important part of the tiles; in particular adjacent tiles cannot have arrows facing each other.

We call \dfn{substitution} a function that maps each tile to a patch of tiles, see Figure~\ref{fig:htpf-subst}.
Note that in Figure~\ref{fig:htpf-subst}, only the tiles in full lines and with labels are in the image, and the shapes in dotted lines are called \dfn{forced-neighbour tiles}.
The substitution is then extended to patches of tiles by applying it separately to each tile and gluing together the obtained patches (on the forced-neighbour tiles in our case). For more about geometric and combinatorial substitutions see for example \cite{jolivet}.

A substitution $\substitution$ is called \dfn{well-defined} when for any single tile, the substitution can be iterated infinitely many times, and in doing so, the whole plane is asymptotically covered.
In this case, we call \dfn{admissible} any pattern that is a subpattern of $\sigma^k(t)$ for some integer $k$ and tile $t$.
An infinite tiling $\tiling$ is called \dfn{valid} for the substitution if all of its patterns $\pattern \sqsubset \tiling $ are admissible.

We call \dfn{scaling} a substitution $\sigma$ for which there exists a scalar $\lambda > 1$ such that for any admissible pattern $\pattern$, we have $\lambda \support(\pattern) \subset \support(\sigma(\pattern))$, where $\support(\pattern)$ is the support of the pattern (the union of its tiles).

\begin{thm}[HTPF substitution \cite{hat}]~
  \begin{itemize}
  \item The $HTPF$ substitution from Figure~\ref{fig:htpf-subst} is well-defined and scaling;
  \item Any $HTPF$ tiling which is valid for the substitution can be decomposed into a Hat tiling (see Figure~\ref{fig:htpf-clusters}) with the same symmetries \cite[proof of Theorem~2.1]{hat}.
    \end{itemize}
  \label{thm:cite-hat-substitution}
\end{thm}
The original proof that the Hat can tile the whole plane is based on the $HTPF$ substitution.
Note that in \cite{hat}, a second substitution (with clusters $H_7$ and $H_8$) is proposed, but it is less suited to our needs, as it is not scaling.

\begin{figure}[htp]
  \center
  \includegraphics[width=\textwidth]{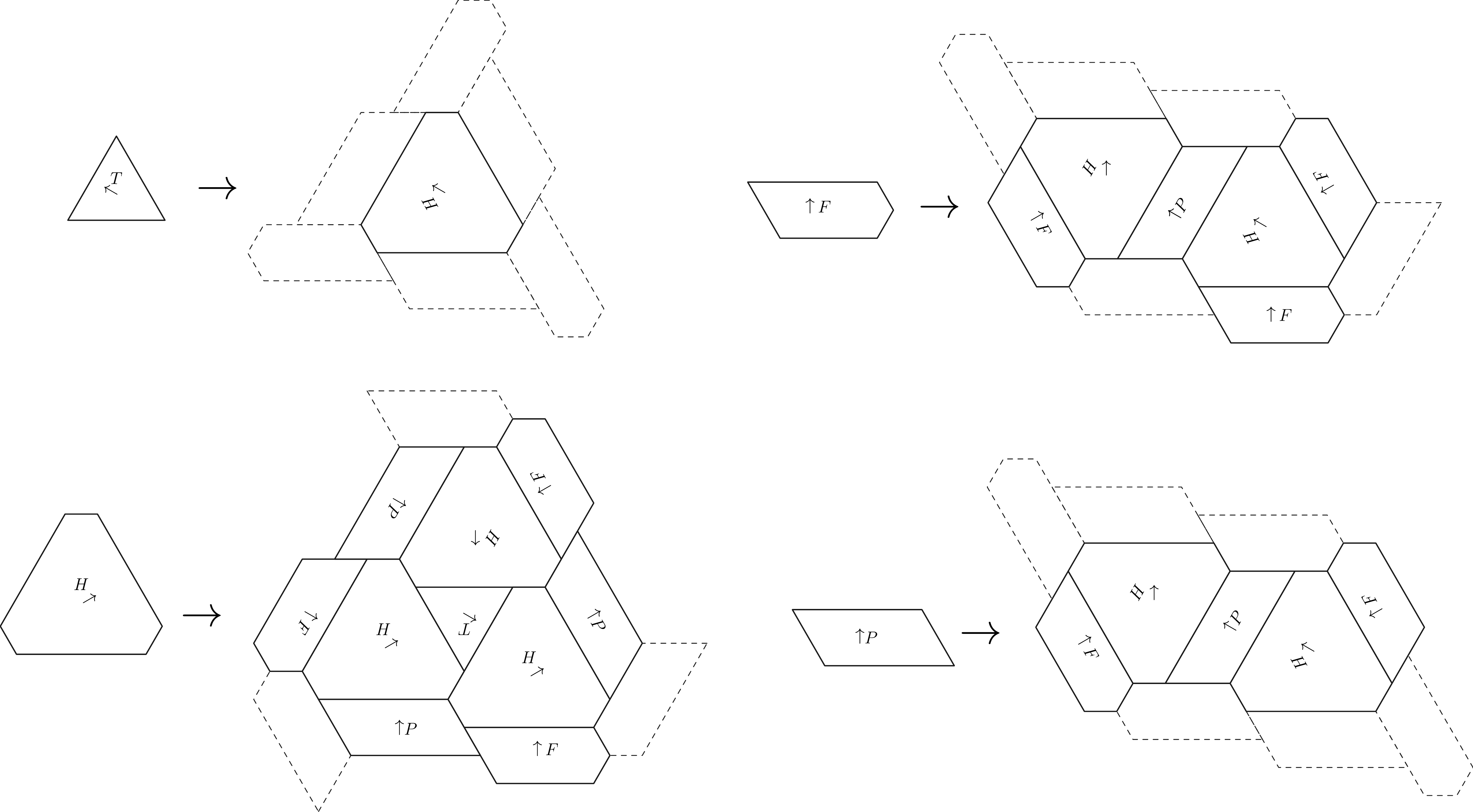}
  \caption{The $HTPF$ substitution. In the image of the tiles, shapes drawn in dotted lines are forced-neighbour tiles.}
  \label{fig:htpf-subst}
\end{figure}

\begin{figure}[htp]
  \center
  \includegraphics[width=.6\textwidth]{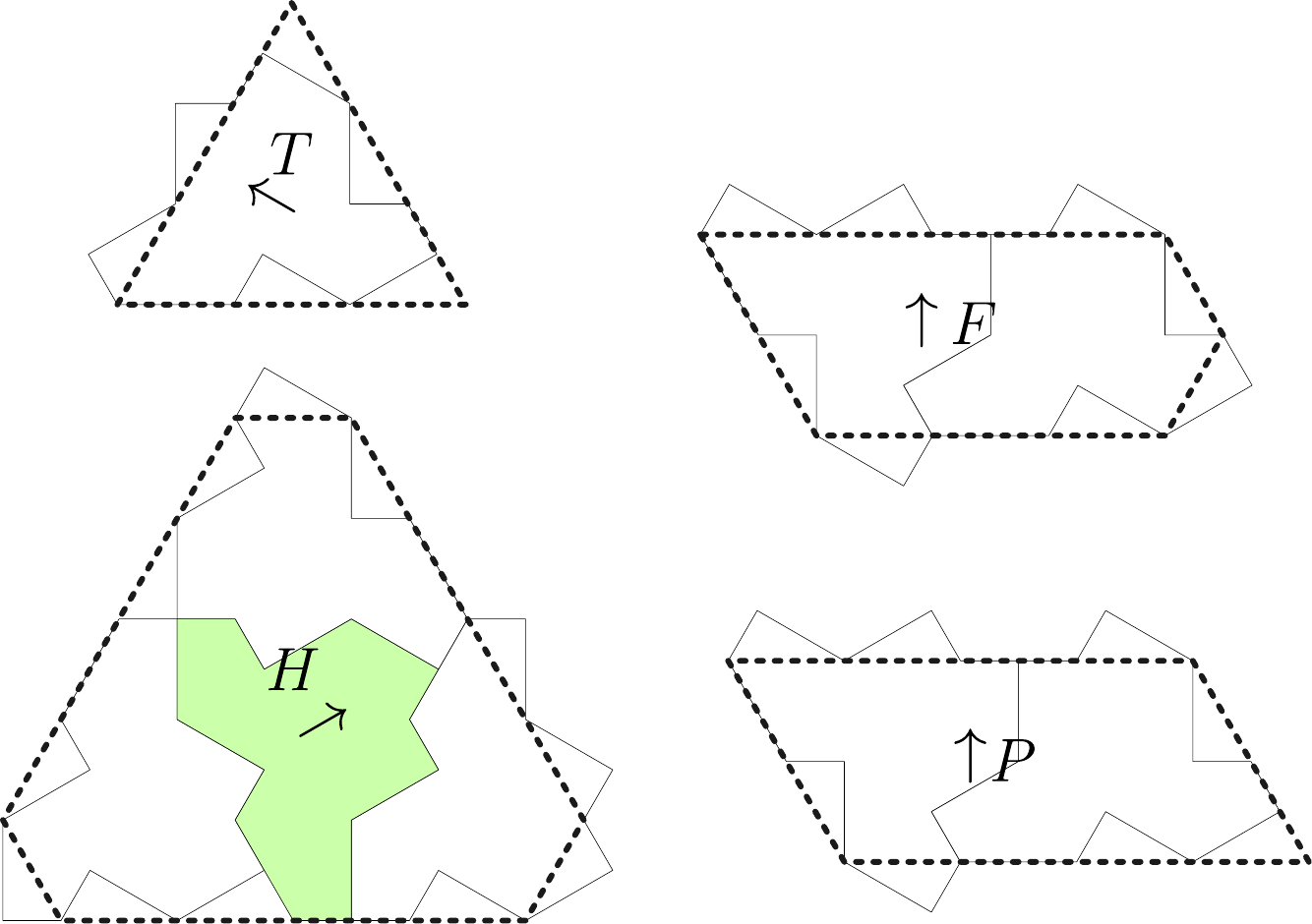}
  \caption{$HTPF$ tiles as clusters of Hats, note that these clusters can also be defined as polysemikites.}
  \label{fig:htpf-clusters}
\end{figure}

\begin{prop}\label{p:rothat}
  There exists a Hat tiling $\tiling$ with cotiler $C$ whose stabilizer is exactly $\{\id,R_3,R_3^2\}$, that is: $\Stab(C)=\sym(\tiling)=\{\id, R_3, R_3^2\}$.
\end{prop}
\begin{proof}
  We use the $HTPF$ substitution to construct a fixpoint (limit) tiling with symmetries $\{\id, R_3, R_3^2\}$.
  We denote by $\substitution$ the $HTPF$ substitution.

  We use the \dfn{fylfot} or $F_3$  pattern consisting of 3 $F$ tiles (see Figure~\ref{fig:fylfot} left) with the central vertex on the origin $0\in \R^2$.
  First note that this pattern is admissible for the substitution as it appears in $\sigma^2(H)$.
  When iterating $\substitution$ on the fylfot, we see that $F_3$ is at the center of $\substitution^2(F_3)$ up to a rotation of angle $\pi$, so that $F_3$ appears exactly at the center of $\substitution^4(F_3)$.
  From this, we deduce that the sequence $(\substitution^{4n}(F_3))_{n\in\mathbb{N}}$ is a increasing (for inclusion) sequence of patches.
  This, together with the fact that the substitution is well-defined and scaling, gives us the existence of a limit tiling $\tiling_{R_3}^{HTPF}$ to this sequence.
  $\tiling_{R_3}^{HTPF}$ is a fixpoint of $\substitution^4$ and is invariant under rotations $R_3$ and $R_3^2$.
  Decomposing the $HTPF$ clusters into Hat tiles in $\tiling_{R_3}^{HTPF}$, we obtain a Hat tiling $\tiling_{R_3}^{Hat}$ which also has symmetries $\{\id, R_3, R_3^2\}$.
  \begin{figure}[htpb]
    \center
    \includegraphics[width=\textwidth]{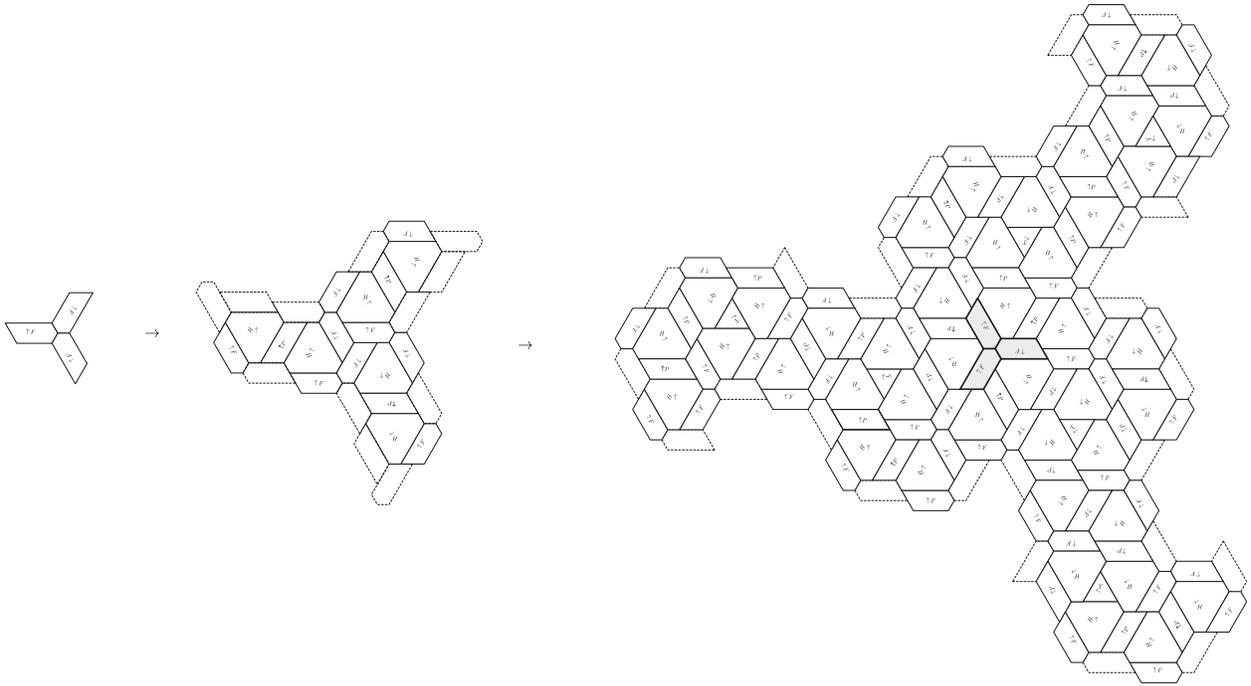}
    \caption{Iterating the $HTPF$ substitution on the fylfot (left) pattern.
      We observe the fylfot at the center of the second iteration.}
    \label{fig:fylfot}
  \end{figure}
\end{proof}

\begin{rem}\label{r:idhat}
  As the $HTPF$ substitution $\substitution$ is sufficiently well-behaved, we can also construct a limit tiling without any non-trivial symmetry.
  
  Indeed, a variation around classical recognizability arguments~\cite{solomyak1998} gives that every valid $HTPF$ tiling has a unique preimage by the substitution.
  Also, no Hat decomposition of $HTPF$ tile $t$ has a $3$-fold symmetry, and this is inherited to supertiles $\sigma^k(t)$, where $k\in\N$.
  The latter property forces any center of $3$-fold symmetry in a tiling to be at the boundary of an infinite-level supertile, and the former states that the boundary of infinite-level supertiles is uniquely defined.

  As the triangle $T$ is at the center of $\sigma^2(T)$ up to a rotation of angle $\tfrac{\pi}{3}$, we construct the fixpoint (limit) tiling $\tiling_{\id}^{HTPF}$ by iterating $\sigma^{12}$ from $T$.
  This tiling has no non-trivial symmetry because no point is at the boundary of an infinite-level supertile.
  We can then decompose $\tiling_{\id}^{HTPF}$ in $\tiling_{\id}^{Hat}$ without any non-trivial symmetry.
\end{rem}

\subsection{The \gat}
\label{subsec:discretehat}
  
  Recall that $\Gamma^{+}\actson\W=\R^2$ and $\Gamma\actson\W$ properly
  discontinuously, and their fundamental domains (up to
  $\lambda$-negligible sets) are the Kite $K$ and the Semikite $K_+$. The
  Hat is a polykite and a polysemikite, which allows to apply the work of
  Section~\ref{s:discr}.

  We get \[\grp[\Gamma^{+},K](\text{Hat})=\{1,R_6,R_6^{-1},R_6^{-2},R_3,R_3^2,R_3R_6^{-1},R_3^2R_6^{-1}\}\subset\Gamma^{+}\] and
  \[\grp[\Gamma,K_+](\text{Hat})=\{1,\alpha,\beta, \gamma,  \alpha\beta,  \beta\alpha,  \beta\gamma,  \gamma\beta, 
    \alpha\beta\alpha, \beta\alpha\beta, \beta\alpha\gamma,
    \beta\gamma\beta, \gamma\beta\alpha, \alpha\beta\alpha\beta,
    \beta\alpha\gamma\beta, \gamma\beta\alpha\beta\}\subset\Gamma;\]
the latter is called \dfn{\gat}.
  \begin{figure}[htp]
    \center
    \includegraphics[width=0.3\textwidth]{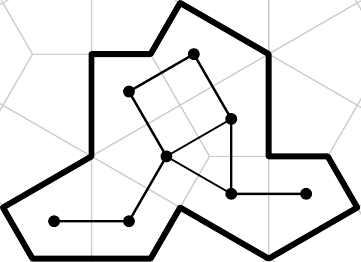} \hspace*{1cm}
    \includegraphics[width=0.3\textwidth]{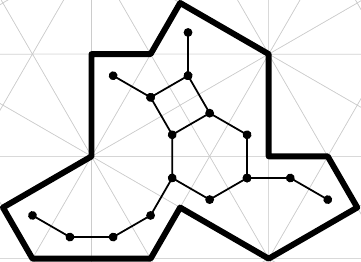}
    \caption{Group tiles induced in $\Gamma^{+}$ (left) and $\Gamma$ (right) by the Hat.}
    \label{fig:hat_dual}
  \end{figure}

Figure~\ref{fig:hat_dual} shows the two sets $\grp[\Gamma^{+},K](\text{Hat})$ and $\grp[\Gamma,K](\text{Hat})$ as subsets of the Cayley graphs of the groups $\Gamma^{+}$ and $\Gamma$.

  \begin{thm}[Aperiodic monotile in $\Gamma$]
    The \gat\ 
    is a mildly aperiodic monotile of the group $\Gamma$.
    More precisely, the possible stabilizers of cotilers of the \gat\ are $\{\id\}$ or conjugated to $\{\id,R_3,R_3^2\}$, both cases being reached.
  \end{thm}
\begin{proof}
  First, we apply Item~\ref{i:ct} of Remark~\ref{r:ct}: as any Hat tiling sits on a Kitegrid, any cotiler of the Hat with respect to $\Isom(\R^2)$ is also (up to a global isometry $g\in \Isom(\R^2)$) a cotiler of the Hat with respect to $\Gamma$.
  
  We then apply Corollary~\ref{c:main-one-k} to obtain that the set of cotilers for the Hat and \gat\ with respect to $\Gamma$ are identical.

  Now, Proposition~\ref{p:aperhat} gives that any possible cotiler stabilizer for \gat\ is a subgroup of $\{\id,R_3,R_3^2\}$.
  By Proposition~\ref{p:rothat} and Remark~\ref{r:idhat}, we obtain that this bound is tight, that is: there exists a cotiler $C$ for the \gat\ with $\Stab(C)=\{\id,R_3,R_3^2\}$, and a cotiler $C'$ with $\Stab(C')=\{\id\}$.
  \begin{figure}[htp]
    \center
    \includegraphics[width=0.6\textwidth]{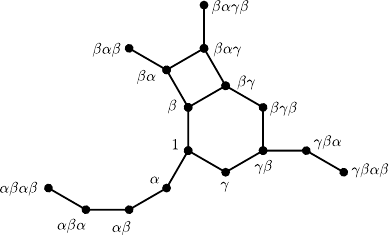}
    \caption{The \gat\ tile: the $\Gamma$-tile induced by the Hat.}
    \label{fig:hat_dual_labels}
  \end{figure}
  \end{proof}

  \begin{rem}
    The set $\grp[\Gamma^{+},K](\text{Hat})$ does not tile the group $\Gamma^{+}$.
    
 Indeed, using the same reasoning, we obtain that the set of cotilers of $\grp[\Gamma^{+},K](\text{Hat})$ is equal to the set of cotilers of the Hat with respect to $\Gamma^{+}$.
 As $\Gamma^{+}$ is a subgroup of the group of direct isometries of $\R^2$ (it contains no reflection), and the Hat does not tile $\R^2$ without reflections \cite{hat}, we obtain that the set of cotilers of $\grp[\Gamma^{+},K](\text{Hat})$ is empty.

 However, using the reflection $\alpha$ around a Kitegrid axis, the set $\{\grp[\Gamma^{+},K](\text{Hat}), \grp[\Gamma^{+},K](\alpha(\text{Hat}))\}$ is a $2$-tile mildly aperiodic tileset in $\Gamma^{+}$, see Figure~\ref{fig:hat_and_reflection}.

 \begin{figure}[htp]
   \center
   \includegraphics[width=0.5\textwidth]{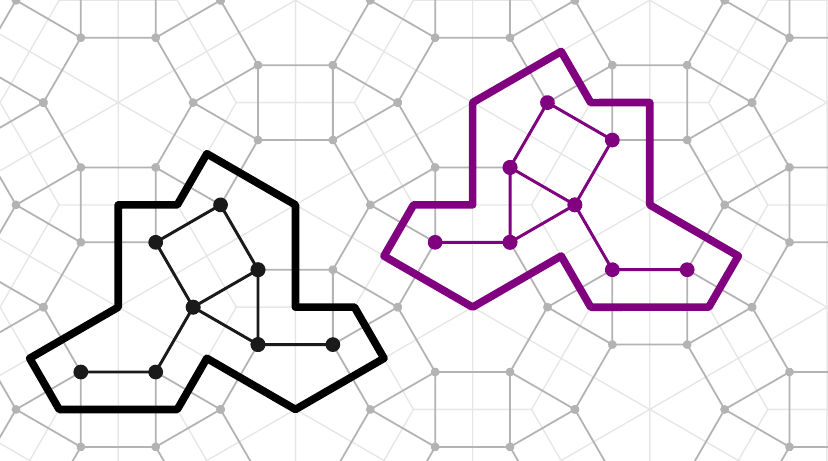}

   \caption{The 2-tile mildly aperiodic tileset induced by the Hat and its symmetric in $\Gamma^{+}$ (up to translation); the two tiles are not $\Gamma^+$-translations of each other.}
\label{fig:hat_and_reflection}
 \end{figure}
  \end{rem}

  The Turtle tile \cite{hat} (Figure~\ref{fig:turtle}) or Hat-Turtle systems \cite{spectre} (Figure~\ref{fig:hat-turtle}) can also be turned group-theoretic. 

  \begin{figure}[htp]
    \center \includegraphics[width=0.3\textwidth]{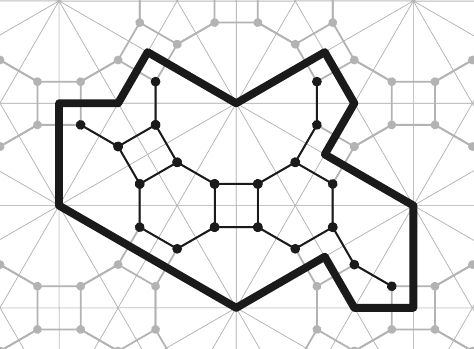}
    \caption{The Turtle tile \cite{hat} induces a second mildly aperiodic monotile in $\Gamma$.}
    \label{fig:turtle}
  \end{figure}

  \begin{figure}[htp]
    \center \includegraphics[width=0.5\textwidth]{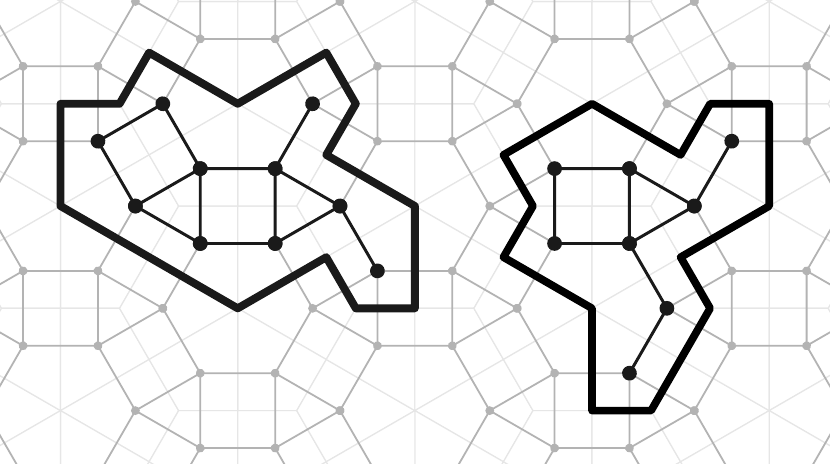}
    \caption{The Hat-Turtle system \cite{spectre} induces a 2-tile mildly aperiodic tileset in $\Gamma^{+}$ (without the reflections).}
    \label{fig:hat-turtle}
  \end{figure}

  Note that the group $\Gamma$ is obtained by adding a reflection to
  the group $\Gamma^{+}$: $\Gamma=\Gamma^{+}\rtimes\langle\alpha\rangle$.
  Thus the Cayley graph of $\Gamma^{+}$ can be recovered from two
  interlaced copies inside the Cayley graph of $\Gamma$.
  Restricting to the \gat: Figure~\ref{fig:discrete_hat_semidirect}
  shows two interlaced copies of $\grp[\Gamma^{+},K](\text{Hat})$ inside it.

  \begin{figure}[htp]
    \center
    \includegraphics[width=0.3\textwidth]{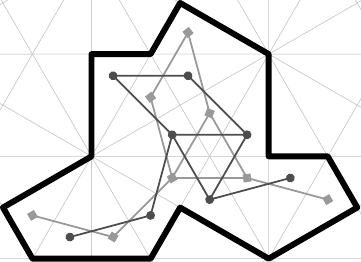}
    \caption{We can partition the vertices of the \gat\ into two copies of the $\Gamma^{+}$-discrete Hat (from Figure~\ref{fig:hat_dual}, left), each of which corresponds to one orientation (even / odd signature of the transformation).}
    \label{fig:discrete_hat_semidirect}
  \end{figure}

  \FloatBarrier
  \section{General version}\label{sec:general}
  In this section, we consider a topological space $\W$ endowed with a nontrivial Borel measure $\lambda$, a group $G$ acting on it by negligibility-preserving self-homeomorphisms, and a finite set $\T$ of positive-measure tiles $T\subset\W$.

  When is the discretization as performed in Section~\ref{s:discr} possible?
  We need a transformation map turning geometric tiles into group tiles.
  The first candidate is the following.
  Given a base point $0\in\W$, one can define its \dfn{orbit map} from $G$ to $\W$: $\orb_0:g\mapsto g\cdot 0$.
  One can build, from any set $T\subset\W$, the set $\grp(T)\defeq\orb_0^{-1}(T)=\sett[G]{g}{g\cdot 0\in T}\subset G$.

  In this section, we give two general statements: one for transitive actions (Subsection~\ref{ss:trans}), and one for discrete actions (Subsection~\ref{ss:discr}, somewhat generalizing Section~\ref{s:discr}).
  
  \subsection{Building tiles in continuous groups}\label{ss:trans}
   Here is a continuous variation on our theorem.
   We include it for completeness and the parallelism with the discrete version.

   We assume in this subsection that $G$ is a topological group endowed with some nontrivial Borel measure.
   A map between measure spaces is \dfn{nonsingular} if the preimage of any negligible subset is negligible.
   It is \dfn{infinite-measure-preserving} if the image of any infinite-measure subset has infinite measure.
   An action of a topological group $G$ on $\W$ is \dfn{proper} if for every compact $K,K'\subset\W$, $\sett[G]g{g\cdot K\cap K'}$ is compact.

   The tile $\grp(T)$ obtained in the group enjoys nice properties.
   \begin{rem}\label{r:nicetiles}For every $0\in\W$,
     \begin{enumerate}
  \item If $\orb_0$ is negligibility-preserving and $T$ has positive measure, then $\grp(T)$ has positive measure;
  \item If $\orb_0$ is {in}finite-measure-preserving and $T$ has finite measure, then $\grp(T)$ has finite measure;
  \item If $G\actson\W$ is proper and $T$ is {relatively compact}, then $\grp(T)$ is compact.
\end{enumerate}\end{rem}

\begin{prop}\label{p:trans}
  Endow $G$ with a measure and chose $0$ such that $\orb_0$ is nonsingular.
  
  If $C$ is a cotiler of $\T$ in $\W$ with respect to $G$, then $C$ 
  is a cotiler of $\grp(\T)$ in $G$, up to syntactically renaming each $T$ into $\grp(T)$.
\end{prop}
\begin{proof}
  By nonsingularity and the covering condition for tilings, for $\lambda$-almost all $h\in G$, there exists a unique $(g,T)\in C$ such that {$h\cdot 0 \in g\cdot T$, which is equivalent to $g^{-1}h\in\grp(T)$, and to $h\in g\grp(T)$}.
  This proves that
  \[
    G \equiv_\lambda\bigsqcup_{(g,T)\in C} g\grp(T).
\popQED  \]
\end{proof}

\begin{thm}[Main result; continuous setting]\label{t:maincont}
  If $G\actson\W$ is transitive and $T$ is any tile, then $\geo(\grp(T))=T$.
  Moreover, if $\T$ is a finite tileset, $\mu$ is a measure on $G$ which makes $\orb_0$ nonsingular and negligibility-preserving, then $\T$ and $\grp(\T)$ have the same cotilers (up to syntactically renaming each $T$ into $\grp(T)$).
\end{thm}
\begin{proof}
  $\geo(\grp(T))={G\cdot0\cap T}$.
  If the action is transitive, then this is equal to $T$.
\\ Now by Proposition~\ref{p:trans}, the cotilers of $\T$ are (up to name change) cotilers of $\grp(T)$.
  Conversely, if $G\equiv_\mu\bigsqcup_{(g,\grp(T))\in C}g\grp(T)$, then $G\cdot0\equiv_\lambda\bigsqcup_{(g,\grp(T))\in C}g\cdot\geo(\grp(T))$.
  This gives $\W\equiv_\lambda\bigsqcup_{(g,\grp(T))\in C}g\cdot T$.
\end{proof}

\subsection{Building tiles in discrete groups}\label{ss:discr}
Rather than tiling with respect to continuous groups (often Lie groups)
such as $\Isom(\R^d)$ like in the previous subsection, we are now
interested in countable subgroups (like lattices over Lie groups), for which the orbit map is generally singular; hence Theorem~\ref{t:maincont} does not apply.

We propose two constructions for carrying a tiling inside the group.

\subsubsection{Using the orbit map}

In this subsection, $G$ is a countable group, endowed with discrete topology and counting measure.
Although Theorem~\ref{t:maincont} does not apply, we can compensate by a suitable choice of base point.

An immediate remark is that our discretization map is relevant for $\equiv_\lambda$-classes of subsets.
\begin{rem}\label{r:equiv}
  If $0\notin G\cdot(T\Delta T')$ (where $\Delta$ stands for the symmetric difference), then $\grp(T)=\grp(T')$.\\
  In particular, if $T\equiv_\lambda T'$, then for $\lambda$-almost every $0$, $\grp(T)=\grp(T')$.
\end{rem}

One says that the action $G\actson\W$ is \dfn{$\lambda$-properly discontinuous} if for every compact $K\subset\W$, only finitely many $g\cdot K$ intersect $K$ $\lambda$-non-trivially.
Note that this property is equivalent to: for every relatively compact (that is, included in a compact) $K,K'\subset\W$, only finitely many $g\cdot K$ intersect $K'$ $\lambda$-non-trivially.
Any properly discontinuous (in the sense that for every compact $K\subset\W$, only finitely many $g\cdot K$ intersect $K$) action is $\lambda$-properly discontinuous for every measure $\lambda$.

Our transformation map is nice in the sense that the tiles inherits some simplicity properties from the action.
\begin{prop}\label{p:fini}
  Let $T$ be {relatively compact}
  and $G\actson\W$ be properly discontinuous (resp, 
  $\lambda$-properly discontinuous), then $\grp(T)$ is finite for every (resp, $\lambda$-almost every)  $0\in\W$.
\end{prop}
\begin{proof}
  $\grp(T)$ can be written as $\sett[G]g{T\cap g\cdot\{0\}\ne\emptyset}$, which is finite if the action is properly discontinuous. 
  If the action is $\lambda$-properly discontinuous, 
  this implies that $\U\defeq\sett[G]g{\lambda(T\cap g\cdot T)>0}$ is finite.
  By definition and countable additivity, $R\defeq G\cdot\bigcup_{g\notin\U}T\cap g\cdot T$ has measure $0$.
  For every $0\notin R$, $\grp(T)\subset\U$, which is finite.
\end{proof}

\begin{prop}\label{p:discpav}
  Let $G$ be a countable group, endowed with the counting measure.
  If $C$ is a cotiler of $\T$ in $\W$,
then for almost every $0\in\W$,
  $C$ is a cotiler of $\grp(\T)$ (with respect to $G$) in $G$ (up to syntactically renaming each $T$ into $\grp(T)$).
\end{prop}
\begin{proof}
  Let us show the result for $0$ in the set
  \[\bigcap_{h\in G}h(\bigcup_{(g,T)\in C}g\grp(T)\setminus\bigcup_{(g,T)\ne(g',T')\in C}g\grp(T)\cap g'\grp{(T')}).\]
  By definition of tilings and countable additivity of the measure, this set has full measure.

  By construction, for every $(g,T),(g',T')\in C$, $g\grp(T)\cap g'\grp(T')=\emptyset$, unless $(g,T)=(g',T')$; and for every $h\in G$, there exists $(g,T)\in C$ such that $h\cdot0\in g\cdot T$, that is to say $h\in g\grp(T)$.
\end{proof}

\subsubsection{Decomposing tiles}

We now propose an alternative approach to the orbit map. 

We assume that $G$ is a countable group acting cocompactly $\lambda$-properly discontinuously on $\W$ and that $\T$ is a finite set of relatively compact tiles.

Let $F$ be a relatively compact fundamental domain.
For $x\in F$, consider the set $\chi(x)\defeq\sett[G\times\T]{(g,T)}{x\in g\cdot T}$ of tiles that could cover $x$ according to some transformation.
We also define the \dfn{cell} $K_x\defeq\chi^{-1}(\{\chi(x)\})\subset F$ of $x$.

Let $\K\defeq\sett{K_x}{x\in F,\lambda(K_x)>0}$.
\begin{lem}\label{l:tfini}~
  \begin{enumerate}
  \item\label{i:kfini} Every $K\in\K$ has the form $K=\chi^{-1}(\{\chi(x)\})$ for some finite (possibly empty) $\chi(x)$.
  \item\label{i:kkfini} $\K$ is finite.
  \item\label{i:part} $\bigsqcup_{K\in\K}K\equiv_\lambda F$.
\end{enumerate}\end{lem}
\begin{proof}
  Since the action is $\lambda$-properly discontinuous and the tiles are relatively compact, the set $S_T$ of $g\in G$ such that $g\cdot T$ intersects $F$ $\lambda$-nontrivially is finite, for each $T\in\T$.
  The set $S\defeq\bigsqcup_{T\in\T}S_T\times\{T\}$ is still finite, and by construction, $\forall(g,T)\notin S$, $\lambda(F\cap g\cdot T)=0$.
  Let $F'\defeq\sett[F]x{\chi(x)\subset S}$.
  We have:
  \begin{align*}
    F\setminus F'&=\bigcup_{(g,T)\notin S}\sett[F]x{(g,T)\in\chi(x)}
    \\&=\bigcup_{(g,T)\notin S}F\cap g\cdot T.
  \end{align*}
  By definition of $S$ (and countable additivity), $\lambda(F\setminus F')=0$.
  We get that if $\lambda(K_x)>0$, then $\chi(x)\subset S$.
  This gives Item~\ref{i:kfini}, and since each $K\in\K$ has the form $\chi^{-1}(\chi(x))$ for some $\chi(x)\in S$, we get Item~\ref{i:kkfini}.

  By definition of $K_x$ as preimages, we know that they are all disjoint.
  Now
  \begin{align*}
    F\setminus\bigsqcup_{K\in\K}K\quad&\subset \quad (F\setminus F')\cup\bigcup_{S'\subset S,\lambda(\chi^{-1}(S'))=0}\chi^{-1}(S').
  \end{align*}
  By finite additivity, this set has measure $0$.
  This gives Item~\ref{i:part}.
\end{proof}

\begin{lem}\label{l:polyk}
  Every tile $T\in\T$ is poly-$\K$.
\end{lem}
\begin{proof}
  By Item~\ref{i:part} of Lemma~\ref{l:tfini}, $T\subset\W\equiv_\lambda\bigsqcup_{(g,K)\in G\times\K}g\cdot K$.
  Moreover, if $g\cdot K$ nontrivially intersects $T$, this means that $K$ nontrivially intersects $g^{-1}\cdot T$: $(g^{-1},T)\in\chi(K)$.
  Hence $K\subset g^{-1}\cdot T$, and $g\cdot K\subset T$.
  Finally, $\lambda$-proper discontinuity gives that the set $\grp[\K](T)$ of $(g,K)$ such that $g\cdot K$ intersects $T$ nontrivially is finite.
  We obtain $T\equiv_\lambda\bigsqcup_{(g,K)\in\grp[\K](T)}g\cdot K$.
\end{proof}

\begin{thm}[Main result; discrete setting]\label{t:main}
  Assume that $G$ is countable and acts on $\W$ cocompactly $\lambda$-properly discontinuously, and that $\T$ is a finite set of relatively compact subsets of $\W$.
  Then there exists a finite set $\K$ of relatively compact tiles yielding a grid on $\W$ such that the tiles in $\T$ are poly-$\K$s,
  and $\T$ and $\grp[\K](\T)$ have the same cotilers (up to syntactically renaming each $T$ into $\grp[{\K}](T)$).
   \end{thm}
   Consequently, if $G$ is countable and $G\actson\W$ is $\lambda$-properly discontinuous and admits some (weakly, mildly, strongly) aperiodic set of relatively compact tiles in $\W$, then one can find some (weakly, mildly, strongly) aperiodic set of finite tiles in $G$. 
   \begin{proof}
     This comes directly from Lemma~\ref{l:tfini}, Lemma~\ref{l:polyk}, and Theorem~\ref{t:maink}.
   \end{proof}
   
   In the previous statement, the transformation $\grp[\K](\T)$ depends on the set $\K$ of tiles that is built.
   It is posible to link this discretization with the discretization $\grp$, but in general, one gets only one inclusion between the cotiler sets.
   This is stated in the following proposition, which applies, for instance, just after Theorem~\ref{t:main}, to cocompact $\lambda$-properly discontinuous actions.
\begin{prop}
  If $\K$ yields a grid, then
  for $\lambda$-almost every $0\in\W$,
  there exists a unique $(g,K)\in G\times\K$ such that $0\in g\cdot K$.
  Then if $T$ is poly-$\K$, then $\grp(T)=\grp[g\cdot K](T)$.
  In particular, if $0_K\in\K$ is such a generic point for each $K\in\K$, then \[T\equiv_\lambda\bigsqcup_{K\in\K}\grp[0_K](T)\cdot K.\]
  In particular, if $\T$ is a set of poly-$\K$, then (up to syntactically renaming each $T$ into $\grp(T)$):
  \[\ct(\T)=\bigcap_{K\in\K}\ct(\grp[0_K](\T))=\bigcap_{0\in\W'}\ct(\grp(\T)),\]
  for some full-measure subset $\W'\subset\W$.
\end{prop}
\begin{proof}
  Let $R\defeq G\cdot\bigcup_{\lambda(g\cdot K\cap T)=0}(g\cdot K\cap T)$.
  By countable additivity, $\lambda(R)=0$.
  By construction, for every $0\in g\cdot K\setminus R$, $h\cdot0\in T\iff\lambda(hg\cdot K\cap T)>0$, which means that $\grp(T)=\grp[g\cdot K](T)$.
  We conclude by Theorem~\ref{t:maink}.
\end{proof}


  \subsection{Crystallographic groups}
  
  From the Hat and our constructions above, we get a finite monotile
  that tiles mildly aperiodically the group $\Gamma$. Moreover,
  the dual tiling of the Semikitegrid is an Archimedean tiling and
  its $1$-skeleton can be labelled to become a Cayley graph of
  $\Gamma$. Thus, the geometric propreties of the Hat carry
  towards the monotile of $\Gamma$: it is finite and connected.

  This highly favorable context is that of crystallographic groups,
  which we explore in this subsection.

  A \dfn{crystallographic group} $G$ is a group acting properly
  discontinuously and cocompactly on a Euclidean space $\W=E$.
  
  From the classical work of Bieberbach~\cite{bieberbach}, the subgroup $L$ of the
  translations in $G$ is generated by $n$ linearly independent
  translations, is isomorphic to $\Z^n$ and has finite index in
  $G$.  For such a crystallographic group $G$, there
  exists a compact convex polytope $K$ {(fundamental domain)} that yields a grid.
  From
  Corollary~\ref{c:main-one-k} a poly-$K$ monotile $T$ is strongly
  (resp. mildly) aperiodic with respect to $G$ if, and only if,
  the monotile $\grp[K](T)$ is strongly (resp. mildly) aperiodic in
  $G$.
  {
  Note that Theorem~\ref{t:main} builds a refinement of that grid (yielded by the fundamental domain), starting from any given set $\T$ of tiles (also, the setting is more general, and may include non-Euclidean spaces).
}

  In dimension $n=2$, following Grünbaum and
  Shephard~\cite[Statement 4.3.1]{gshephard}, the crystallographic grid is {topologically equivalent to} the
  Laves tiling of an Archimedean tiling. The finite set
  $S = \{s\in G\ |\ s\cdot K\text{ and }K\text{ share a hyperface}\}$
  generates $G$. Labelling its edges, the
  $1$-skeleton of the dual Archimedean tiling is the Cayley graph of
  $(G,S)$ (there is a subtlety regarding orientation of edges:
  $S$ contains $s^{-1}$ for each of its element $s$; if such an $s$ is
  a reflection, then one needs not choose an orientation).
  The geometric properties of $T$ translate to geometric
  properties of $\grp[K](T)$; for instance, if $T$ is connected, then
  $\grp[K](T)$ is connected.

  We warn that we only consider tilings by $T$ with respect to
  $G$, that is to say the cotilers are inside $G$. It may
  happen that using more isometries $G\leq G'$ allows new
  tilings by $T$. This situation happened for the Hat: adding a
  reflection $\alpha$ to $G=\Gamma^{+}$, we studied the group $G'=\Gamma$ and
  the Semikitegrid. The Hat does not tile with respect to $\Gamma^{+}$, but
  tiles (aperiodically) with respect to $\Gamma$. For the Hat, no
  other isometry can be added: two adjacent copies of the Hat have to
  share edges and thus have to be drawn on the same Semikitegrid.

  In this context, we raise the following two questions:

  \begin{ques}
    For which crystallographic groups does there exist a mildly aperiodic monotile?

    For the Laves tiling of which Archimedean tiling does there exist a mildly aperiodic poly-$K$ monotile?
  \end{ques}

  \section{Conclusion}
  We are interested in the following metaquestion:
  \begin{ques}
    Which settings $G\actson\W$ admit aperiodic monotiles (or aperiodic finite tilesets)?
    Which countable groups $G$ (with translation action)?
  \end{ques}
After the big 2023 breakthroughs on monotileability, our work brings:
a general framework to compare the results,
a tool to transfer geometric tiles into group tiles,
and an explicit construction for an original aperiodic monotile in a very small group (in the sense that it is close to a group where such monotiles do not exist).
The result is comparable to \cite[Theorem~1.3]{gtao2}, yet more explicit.
On the other hand, the direct product is replaced in our construction by some semi-direct product, which makes our group not a quotient of some higher-dimensional grid.
This prevents to use the lifting operation from \cite{lift} (or even \cite{aperlift}) in order to recover (an explicit version for) \cite[Corollary~1.5]{gtao2}, building a weakly aperiodic monotile for some $\Z^d$ (and some $\R^d$).
It is not clear how to do a similar operation for semi-direct products.

Expliciting the stabilizers also quantifies somehow how far we are from strong aperiodicity.
A great improvement would be to obtain strong aperiodicity (in the sense that the tilings have no symmetry at all, instead of the order-3 $R_3$ for our tile).
One strategy could be to start modify our existing \gat\ tile into something that breaks the rotation (while still tiling), up to finitely extending the group $\Gamma$.

Another strategy would be to apply our \emph{groupification} to a geometric strongly aperiodic monotile.
But, though this is not emphasized in the geometric literature (where aperiodicity simply involves translations), it is surprising to notice that strongly aperiodic monotiles are unknown in all settings, including both the geometric and the group settings.

In particular, the Spectre from \cite{spectre} also admits tilings with stabilizer $\{\id,R_3,R_3^2\}$ (which are combinatorially equivalent to tilings by the Hat).
Anyway, though it is tempting to think of the Spectre as having fewer symmetries, because it tiles with respect to $\Isom^+(\R^2)$ (unlike the Hat), 
on the contrary, the cotiler group of any tiling is much bigger for the Spectre than for the Hat, in the sense that the action is minimal (the orbits are dense), as seen by the fact that the tiling does not lie over an Archimedean tiling.
This is also true for other famous examples of tilesets, like Penrose's.
\begin{ques}
    Does there exist a geometric monotile whose cotilers have trivial stabilizers?
    Does there exist one which is union of tiles from an Archimedean tiling?
  \end{ques}

  Another remark is that our Theorem~\ref{t:maink} is an equivalence, which could in theory lift some results from the group world to the geometric world: for instance, if one knows that $G$ admits no aperiodic monotile, then no space $\W$ over which it acts properly discontinuously admits an aperiodic monotile with respect to $G$.
  An example is that $\R^2$ cannot have aperiodic monotiles with respect to $\Z^2$, thanks to the discrete result from \cite{bhattacharya}, but this is already folklore.

The link formalized in this paper could help formalize the parallelism between the theory of FLC tilings and symbolic dynamics that can be read in many results with similar flavor.

Finally, let us note that aperiodic tilings are an important ingredient for undecidability results (since settings where aperiodic tile sets do not exist imply decidability of tileability and of many properties), both in logical or in algorithmic terms.
The first undecidability results about monotileability were achieved by \cite{gtao3} in some virtually-$\Z^2$ group.
Along history, each work establishing the possibility of aperiodicity in some setting was simultaneous to or soon followed by some work establishing undecidability in this setting (and usually based on similar constructions).
We are thus expecting more of them.
\begin{ques}
  Is it undecidable whether an algebraic polygonal monotile tiles $\R^2$?
  Is it undecidable whether a finite monotile tiles the symmetry group $\Gamma$ of the Semikite grid?
  Does there exist a dimension $d\in\N$ such that it is undecidable whether a finite monotile tiles $\Z^d$?
  \end{ques}

\paragraph{Acknowledgements:}
We thank Ch. Goodman-Strauss and T. Meyerovitch for discussions related to some points of this article, and N. Pytheas Fogg for raising the discussion that brought the original idea.
This work was supported by ANR Project IZES-ANR-22-CE40-0011 by Xpand MathAmSud 240026.

\bibliography{hat}
\bibliographystyle{alpha}

\end{document}

** Hat n Groups

- montrer que si le pavage obtenu dans le groupe est périodique, le géométrique l'était (ça donne une symétrie rotationnelle, mais on peut sans doute l'itérer. connu ?)
- try to reverse Turtle n Hat
- Socolar-Taylor dans la grille hexagonale (recoder les petis carrés en hexagones [Forcing nonperiodicity, Fig 6]) ? ça ne marche pas trop parce que les côtés des hexagones ne sont pas tous étiquetés par les mêmes générateurs.
- citer [[https://scholar.google.fr/citations?user=wca7248AAAAJ&hl=fr&oi=ao][Chaboud]] sur les pavages archimédiens